\documentclass[12pt]{article}

\usepackage{tikz}
\usepackage{graphicx}
\usepackage{verbatim}
\usepackage{array,fullpage,multirow}
\usepackage{latexsym}
\usepackage{enumitem}
\usepackage{amssymb,amsfonts,amsmath,amsthm}
\usepackage{url}
\usepackage[colorlinks=true,linkcolor=blue,citecolor=red]{hyperref}
\usepackage{algorithm}
\usepackage{algorithmic}

\newcommand{\eqdef}{\stackrel{\rm{def}}{=}}
\newcommand{\emptyslot}{\sqcup}

\newcommand{\cc}[1][n]{\{0,1\}^{#1}}
\newcommand{\markedcube}[1][n]{\{0,1,\hat{0}, \hat{1}\}^{#1}}

\newcommand{\dist}{{\sf dist}}

\newcommand{\avgstretch}{{\sf avgStretch}}

\newcommand{\poly}{{\rm poly}}

\newcommand{\N}{{\mathbb N}}

\newcommand{\E}{{\mathbb E}}

\newcommand{\seq}{\subseteq}

\newcommand{\NC}[1]{\mathbf{NC^{#1}}}

\newcommand{\und}{\underline}
\newcommand{\markhat}{{\sf mark}}
\renewcommand{\int}{{\sf int}}
\newcommand{\sig}{{\sf signature}}

\newcommand{\Dict}{{\sf Dictator}}
\newcommand{\Maj}{{\sf Majority}}
\newcommand{\Parity}{{\sf XOR}}

\renewcommand{\Pr}{\mathop{\bf Pr\/}}

\newtheorem{theorem}{Theorem}
\newtheorem{corollary}{Corollary}[section]
\newtheorem{proposition}[corollary]{Proposition}

\newtheorem{lemma}[corollary]{Lemma}
\newtheorem{claim}[corollary]{Claim}

\newtheorem{remark}[corollary]{Remark}
\newtheorem{question}[corollary]{Question}
\newtheorem{problem}[corollary]{Problem}

\begin{document}

\title{On Lipschitz Bijections between Boolean Functions}

\author{
Shravas Rao
	\thanks{
    Courant Institute of Mathematical Sciences,
    New York University.
    Email: {\tt rao@cims.nyu.edu}
    }
\and
Igor Shinkar
	\thanks{
    Courant Institute of Mathematical Sciences,
    New York University.
    Email: {\tt ishinkar@cims.nyu.edu}
    }
}

\maketitle

\begin{abstract}
Given two functions $f,g \colon \cc \to \{0,1\}$
a mapping $\psi \colon \cc \to \cc$ is
said to be a \emph{mapping from $f$ to $g$} if
it is a bijection and $f(z) = g(\psi(z))$ for every $z \in \cc$.
In this paper we study Lipschitz mappings between boolean functions.

Our first result gives a construction of a $C$-Lipschitz mapping from the $\Maj$
function to the $\Dict$ function for some universal constant $C$.
On the other hand, there is no $n/2$-Lipschitz mapping in the other direction,
namely from the $\Dict$ function to the $\Maj$ function.
This answers an open problem posed by Daniel Varga in the paper of
Benjamini~et~al. (FOCS 2014).

We also show a mapping $\phi$ from $\Dict$ to $\Parity$ that is
3-local, 2-Lipschitz, and its inverse is $O(\log(n))$-Lipschitz,
where by $L$-local mapping we mean that each output bit of the mapping
depends on at most $L$ input bits.

Next, we consider the problem of finding functions such that
any mapping between them must have large \emph{average stretch},
where the average stretch of a mapping $\phi$ is defined as
$\avgstretch(\phi) = \E_{x,i}[\dist(\phi(x),\phi(x+e_i)]$.
We show that any mapping $\phi$ from $\Parity$ to $\Maj$ must satisfy
$\avgstretch(\phi) \geq c\sqrt{n}$ for some absolute constant $c>0$.
In some sense, this gives a ``function analogue'' to the question of
Benjamini~et~al. (FOCS 2014), who asked whether there exists a set
$A \seq \cc$ of density 0.5 such that any bijection from
$\cc[n-1]$ to $A$ has large average stretch.

Finally, we show that for a random balanced function $f\colon \cc \to \cc$
with high probability there is a mapping $\phi$ from $\Dict$
to $f$ such that both $\phi$ and $\phi^{-1}$ have constant average stretch.
In particular, this implies that one cannot obtain lower bounds on average
stretch by taking uniformly random functions.

\end{abstract}

\thispagestyle{empty}
\pagebreak


\section{Introduction}\label{sec:intro}

Given two functions $f,g \colon \cc \to \{0,1\}$
a mapping $\psi \colon \cc \to \cc$ is said to be a \emph{mapping from $f$ to $g$}
if $\psi$ is a bijection and for every $z \in \cc$ it holds that $f(z) = g(\psi(z))$.
Of course, if $\E[f] = \E[g]$, then there are many mappings from $f$ to $g$,
and we can further ask whether there are ``simple'' mappings from $f$ to $g$,
where ``simple'' can mean, for example, that $\psi$ is
computable by a small circuit, or has some other nice structure.
In this paper we ask about the existence of Lipschitz mappings between
some well studied boolean functions, including the functions
$\Dict$, $\Maj$, $\Parity$, and a uniformly random balanced function.

As a first example suppose that $f$ is obtained from $g$ by renaming the coordinates.
Then, trivially, there is a $1$-Lipschitz mapping from $f$ to $g$,
which simply permutes the coordinates, and in particular, each output bit of the mapping
depends on exactly one input of its input bits.

In some sense existence of a Lipschitz mapping between $f$ and $g$
implies some similarity between them because
such a mapping induces
\begin{itemize}
\item a Lipschitz bijection from $f^{-1}(0)$ to $g^{-1}(0)$,
\item a Lipschitz bijection from $f^{-1}(1)$ to $g^{-1}(1)$, and
\item a Lipschitz mapping from the cut in $\cc$ defined by $f$
to the cut defined by $g$.
\end{itemize}

Below we summarize the results shown in this paper.

\subsubsection*{Bijections between $\Dict$ and $\Maj$}

It is a recurring theme in the analysis of boolean function that the $\Dict$ function
and the $\Maj$ function are in some senses, opposites of one another.
For example, the Majority is Stablest theorem~\cite{MOO10} states
that if the noise stability of a function significantly deviates from the
noise stability of $\Maj$, the function must have an influential coordinate,
and hence is non-trivially correlated with the corresponding $\Dict$ function.
Another example  is the theorem of Bourgain~\cite{Bourgain}
(see also a recent improvement by Kindler and O'Donnell~\cite{KKO})
saying that if the Fourier transform of a function deviates in an appropriate
sense from that of $\Maj$, then the function can be approximated
by a {\em junta}, i.e., essentially depends on a small number of coordinates,
which also implies some correlation with a $\Dict$ function.

Motivated by questions related to lower bounds on sampling by low-level
complexity classes~\cite{Vi12}, Lovett and Viola~\cite{LV11} suggested
to further explore the differences between the two function,
and asked whether it is true that any bijective mapping
$\phi \colon \Dict^{-1}(1) \to \Maj^{-1}(1)$ must have a large average stretch,
where by average stretch we refer to the quantity
\[
	\avgstretch(\phi) = \E_{x \sim y \in \Dict^{-1}(1)}[\dist(\phi(x),\phi(y))],
\]
with $x \sim y \in \Dict^{-1}(1)$ denoting a random edge in $\{0,1\}^n$ such that $x_1 = y_1 = 1$.

The question has been answer negatively in~\cite{BCS-bilip} in a stronger sense,
where there was shown an explicit bi-Lipschitz bijection that maps $\{0,1\}^{n-1}$
to the upper half of $\cc$ (or equivalently a mapping from $\Dict^{-1}(1)$ to $\Maj^{-1}(1)$).
In the same paper the following problem has been raised.

\begin{problem}\label{problem:daniel varga}
Let $n$ be odd. Is there a bi-Lipschitz bijection
$f \colon \cc \to \cc$ that maps the half cube to the Hamming ball?
In other words, is there a bi-Lipschitz bijection
between $\Maj$ to $\Dict$?
\end{problem}

It is easy to see that there is no $C$-Lipschitz bijection
from $\Dict$ to $\Maj$ for $C < n/2$. Indeed for any bijection $\phi$
from $\Dict(x) = x_1$ to $\Maj$ consider $x \in \{0,1\}^n$ such that $\phi(x) = (1,1,\dots,1)$,
and let $y = x - e_1$. Then, the weight of $\phi(y)$ must be at most $n/2$
since $\Maj(\phi(y)) = \Dict(y) = 0$, and thus $\dist(\psi(x),\psi(y)) \geq n/2$.

In the other direction the answer is not as obvious, and we resolve this question in this paper.
Specifically, we prove the following theorem in Section~\ref{sec:maj-to-dict}.

\begin{theorem}\label{thm:maj-dict}
For all odd integers $n \in \N$ there exists a $C$-Lipschitz bijection
$\psi \colon \cc \to \cc$ from $\Maj$ to $\Dict$, where $C \in \N$ is some absolute constant.
\end{theorem}

As mentioned above, it has been shown in~\cite{BCS-bilip} that there exists a
bi-Lipschitz bijection that maps the upper half of $\cc$
to $\{0,1\}^{n-1}$. Therefore, there exists a bijection from $\cc$ to $\cc$
that maps the upper half of $\cc$ to $\{x \in \cc \colon x_1=1\}$,
and maps the lower half of $\cc$ to $\{x \in \cc \colon x_1=0\}$,
that is Lipschitz on the upper half of the hypercube, and on the lower half of the hypercube.
However, it was not clear how to ``stitch'' these two bijections so that
the endpoints of the edges in the middle layer will also be mapped close to each other.
Theorem~\ref{thm:maj-dict} says that this is indeed possible.

\subsubsection*{Bijections between $\Dict$ and $\Parity$}

We further study the notion of mappings between boolean function
by studying mappings between the $\Dict$ function and the $\Parity$ function.
For this question the well known mapping
$\phi(x_1,x_2,\dots,x_n) = (x_1+x_2,x_2+x_3,\dots,x_{n-1}+x_n,x_n)$
is clearly a bijection from $\Dict(x) = x_1$ to $\Parity$.
Note, however, it is not bi-Lipschitz, as flipping the $k$th bit
in the output changes it preimage in the first $k$ coordinates.
That is if $y=\phi(x)$, then $\phi^{-1}(y+e_k) = x + \sum_{i=1}^k e_i$.

In fact, it is not difficult to come up with a bi-Lipschitz mapping from $\Dict$ to $\Parity$.
Indeed, define $\psi(x_1,x_2,\dots,x_n) = (\Parity(x),x_2,\dots,x_n)$,
It is easy to check that $\psi$ is indeed a $2$-bi-Lipschitz mapping
from $\Dict(x) = x_1$ to $\Parity$.
It makes sense, however, to ask for more, namely, does
there exist a bi-Lipschitz mapping from $\Dict$ to $\Parity$
that is in $\NC0$, i.e., each of its output bits depends only
on a constant number of input bits.%
\footnote{Note that the inverse mapping, namely a bijection from $\Parity$ to
$\Dict$ cannot be local since the dictating output coordinate must be the parity of all
input bits, and thus depend on all of them.}

We prove the following theorem in Section~\ref{sec:dict-to-xor}.
\begin{theorem}\label{thm:dict-xor}
    There exists a Lipschitz mapping $\phi$ from $\Dict$ to $\Parity$
    such that each of its output bits depends on at most 3 input bits,
    $\phi$ is 2-Lipschitz, and its inverse $\phi^{-1}$ is $O(\log(n))$-Lipschitz.
    Furthermore, the mapping $\phi$ is a linear operator over $GF(2)$.
\end{theorem}

\subsubsection*{Bijections between $\Maj$ and $\Parity$}

In the paper~\cite{BCS-bilip} the authors asked whether there exists a
subset $A \subset \cc[n+1]$ of density $1/2$ such that
any bijection from $\cc$ to $A$ must map endpoints
of many edges of the hypercube far apart.
Specifically, for a mapping $\phi \colon \cc \to A$ they define
the \emph{average stretch of $\phi$} as
$\avgstretch(\phi) = \E_{x \in \cc, i \in [n]}[\dist(\phi(x),\phi(x+e_i)]$
and pose the following problem.

\begin{problem}\label{problem:unbounded stretch}
Is there a subset $A \subset \cc[n+1]$ of density $1/2$ such that
any bijection $\phi \colon \cc \to A$ has
$\avgstretch(\phi) = \omega(1)$.
\end{problem}

We remark that we are not aware of the existence of a subset $A \subset \cc[n+1]$ of density $1/2$
such that any bijection $f \colon \cc \to A$ has $\avgstretch(f) > 2.1$,
and we find this open problem very interesting.
It also makes sense to relax Problem~\ref{problem:unbounded stretch}
to an appropriate 2-set version, where we ask for two sets $A,B \in \{0,1\}^n$
of density $1/2$ such that any bijection $f \colon A \to B$ has large average stretch
in the appropriate sense.

Below we give a positive answer to the ``function analogue'' of this question.
Specifically, we show that any bijection from $\Parity$ to $\Maj$ must have large average stretch.

\begin{theorem}\label{thm:xor-maj}
    Any mapping $\phi \colon \cc \to \cc$ from $\Parity$ to $\Maj$ must satisfy
	\[
	\avgstretch(\phi) \geq c\sqrt{n}
	\]
	for some absolute constant $c>0$.

    On the other hand, there exists a $C$-Lipschitz mapping $\psi \colon \cc \to \cc$ from $\Maj$ to $\Parity$
    for some absolute constant $C$.
\end{theorem}

We prove Theorem~\ref{thm:xor-maj} in Section~\ref{sec:xor-to-maj}.

\subsubsection*{Bijections between $\Dict$ and a random balanced function}

We also show that for a random balanced function $f \colon \{0,1\}^n \to \{0,1\}$
with high probability
there is a mapping $\phi$ from the Dictatorship function to $f$
such that both $\phi$ and $\phi^{-1}$ have constant average stretch.
The proof of Theorem~\ref{thm:rand} appears in Section~\ref{sec:dict-to-rand}.

\begin{theorem}\label{thm:rand}
	Let $f \colon \{0,1\}^n \to \{0,1\}$ be a uniformly random balanced boolean function.
	Then, with probability $1-2^{-2^{\Omega(n)}}$ there exists a
	mapping $\phi \colon \{0,1\}^n \to \{0,1\}$ from $\Dict$ to $f$ such that
	for $1-O(1/n)$ fraction of $x \in \{0,1\}^n$ it holds that $\dist(x,f(x)) \leq 2$.
	In particular, $\phi$ satisfies $\avgstretch(\phi) = O(1)$
	and $\avgstretch(\phi^{-1}) = O(1)$.
\end{theorem}

This implies that for two random balanced functions
with high probability there is a bijective mapping between them
such that both the mapping and its inverse have constant average stretch.

\begin{corollary}\label{cor:rand2}
	Let $f,g \colon \{0,1\}^n \to \{0,1\}$ be two uniformly random balanced boolean function.
	Then, with probability $1-2^{-2^{\Omega(n)}}$ there exists a
	bijection $\phi \colon \{0,1\}^n \to \{0,1\}$ from $f$ to $g$ that
	satisfies $\avgstretch(\phi) = O(1)$
	and $\avgstretch(\phi^{-1}) = O(1)$.
\end{corollary}

Indeed, let $\phi_f, \phi_g$ be bijections given by Theorem~\ref{thm:rand} when
applied on $f$ and $g$ respectively. Then it is easy to see that the composition
of $\phi_g$ with the inverse of $\phi_f$ gives us the desired mapping
$\phi = \phi_g \circ \phi_f^{-1}$.
Indeed, since $\phi_f$ is a bijection, by Theorem~\ref{thm:rand} it
satisfies $\Pr_{x \in \cc}[\dist(\phi_f(x),x) \leq 2]  =1 - O(1/n)$,
and hence
\begin{eqnarray*}
\Pr_{x}[\dist(\phi_g(\phi_f^{-1}(x)),x) \geq 4]
& \leq &
\Pr_{x}[\dist(\phi_g(\phi_f^{-1}(x)),\phi_f^{-1}(x)) \geq 2]
+
\Pr_{x}[\dist(\phi_f^{-1}(x),x) \geq 2] \\
&  = & O(1/n).
\end{eqnarray*}
Therefore, for $O(1/n)$ fraction of the edges it holds that
$\dist(\phi(x),\phi(x+e_i)) \leq 9$. For the remaining $O(1/n)$
fraction of the edges their endpoints are trivially mapped to distance at most $n$
and so the average stretch of $\phi$ is $O(1)$, as required.
\subsection{Notation}\label{sec:notation}

The functions used in this paper are the following.
The function $\Maj \colon \cc \to \{0,1\}$ is defined as
\[
    \Maj(x) =   \begin{cases}
                    1 & \text{ if } \sum_{i=1}^n x_i > n/2 \\
                    0 & \text{ otherwise.}
                \end{cases}
\]
The function $\Dict \colon \cc \to \{0,1\}$ is defined as $\Dict(x) = x_1$,
i.e. its value is dictated by the first coordinate.
The function $\Parity$ is defined as $\Parity \colon \cc \to \{0,1\}$ is defined as $\Parity(x) = \sum_{i=1}^n x_i \pmod 2$.

A mapping $\phi \colon \{0,1\}^n \to \{0,1\}^n$ is said to be {\em $C$-Lipschitz} if
for every $x,y \in \{0,1\}^n$ it holds that
$\dist(\phi(x),\phi(y)) \leq  C \dist(x,y)$, where $\dist(\cdot,\cdot)$
denotes the Hamming distance between the strings.

Note that in order to prove that a mapping $\phi$ is $C$-Lipschitz
it is enough to show that for every edge of the hypercube $(x,x+e_i)$ it holds that
$\dist(\phi(x),\phi(x+e_i)) \leq  C$.

As a relaxation of the notion of being $C$-Lipschitz
define the average stretch of $\phi$ as
$\avgstretch(\phi) = \E_{x,i}[\dist(\phi(x),\phi(x+e_i)]$.
This means that if $\avgstretch(\phi)$ is large then
many edges of the hypercube far apart, while if it is small,
then the endpoints of an average edge are mapped by $\phi$ close to each other.

\section{A Bijection from $\Maj$ to $\Dict$}\label{sec:maj-to-dict}

In this section we prove Theorem~\ref{thm:maj-dict}.
The proof is based on the idea from~\cite{BCS-bilip},
which relies on a classical partition of the vertices of $\cc$ to symmetric chains,
due to De~Bruijn, Tengbergen, and Kruyswijk~\cite{BTK51}, where a symmetric chain
is a path $(c_k, c_{k+1}, \ldots, c_{n-k})$ in $\cc$, such that
each $c_i$ has Hamming weight $i$.

De~Bruijn, Tengbergen, and Kruyswijk~\cite{BTK51} suggested a recursive
algorithm that partitions $\cc$ to symmetric chains. We will follow
the presentation of the partition described in~\cite{LW2001} (see Problem 6E
in Chapter 6), and we shall call it the \emph{BTK partition}.
We describe the partition by specifying for each $x \in \cc$ the chain $C_x$
that contains $x$.

The algorithm is iterative. During the running of the algorithm, every coordinate
of $x$ is either marked or unmarked, where we denote a marked $0$ by $\hat{0}$
and a marked $1$ by $\hat{1}$.
In each step, the algorithm chooses a consecutive pair $10$, marks it
by $\hat{1}\hat{0}$, temporarily deletes it, and repeats the process.
The algorithm halts when no such consecutive pair is left, i.e., the remaining string is of the form
$00\dots01\dots11$. We call this stage of the algorithm the \emph{marking stage},
and denote the marked string by $\markhat(x) \in \{0,1,\hat{0},\hat{1}\}^n$.
Define the \emph{signature of $x$}, denoted by $\sig(x) \in \{0,1,\emptyslot\}$ as follows:
if the $i^{\text{th}}$ bit of $x$ was marked then $\sig(x)_i = x_i$
and otherwise, $\sig(x)_i = \emptyslot$.
Finally, define $C_x$ to be the collection of all strings whose signature is
equal to $\sig(x)$. That is all strings $y$ agree with $x$
in the marked coordinates of $x$, and in the remaining coordinates
$y$ is of the form $00\dots01\dots11$.

For example, consider the string $x = 01100110$. In the first iteration, the
algorithm may mark the third and fourth bits to obtain $01\hat{1}\hat{0}0110$.
Then, the second and fifth bits are marked 0\^{1}\^{1}\^{0}\^{0}110. Lastly,
the rightmost two bits are marked, and we obtain the marked string
$\markhat(x) = 0\hat{1}\hat{1}\hat{0}\hat{0}1\hat{1}\hat{0}$.
Therefore, the signature of $x$ is $\sig(x) = \emptyslot 1100 \emptyslot 10$ and
$C_x = \{ \und{0}1100\und{0}10, \und{0}1100\und{1}10, \und{1}1100\und{1}10 \}$.

Note that although the algorithm has some degree of freedom when choosing
the order of marking the $10$ pair out of possibly many pairs in a given iteration,
the chain $C_x$ is, in fact, independent of the specific choices that were made.
That is, $\sig(x)$ is a function of $x$, and does not depend on the specific
order in which the algorithm performs the marking. An alternative way to see it
is to think of $1$'s as opening parentheses and of $0$'s as closing parentheses,
and then mark all maximal sub-sequences of legal parentheses in the given string $x$.
As a consequence, we may choose the $10$ pairs in any order we wish.
We will use this fact in the proof of Theorem~\ref{thm:maj-dict}.

The key part of the proof is the following lemma.
We remark that this lemma appears implicitly in~\cite{BCS-bilip}.

\begin{lemma}\label{lemma:BTK dist}
	Let $n \in \N$, and let $x,y \in \cc$ be such that $\dist(x,y) = 1$.
	Let $C_x = \{c_k,c_{k+1}\dots,c_{n-k} \}$ and $C_y  = \{c'_{k'},c_{k'+1}\dots,c'_{n-k'} \}$
	be the chains of the BTK partition that contain $x$ and $y$ respectively.
	Then, $\dist(\sig(x),\sig(y)) \leq 3$, where $\dist(\cdot,\cdot)$ denotes
	the Hamming distance between two strings,
	that is, the number of coordinates where the two strings differ.
	In particular this implies that
	\begin{enumerate}
	\item $|k-k'| \leq 1$.
	\item If $c_j \in C_x$ and $c'_{j'} \in C_y$ for some $j \in [k,n-k]$ and $j' \in [k',n-k']$,
			then $\dist( c_j,c'_{j'} ) \leq |j-j'| + 6$.
	\end{enumerate}
\end{lemma}

In particular, if $x \sim y$, then the Hausdorff distance between $C_x$ and $C_y$ is at most $d_H(C_x,C_y) \leq 7$.

\begin{proof}
	Fix $x,y \in \cc$ such that they differ only in the $i$th
	coordinate and $x_i=0$ and $y_i=1$.
	We may perform the marking stage on each of them in three steps:
		\begin{enumerate}
		  \item Perform the marking stage on the prefix of the string of length $i-1$.
		  \item Perform the marking stage on the suffix of the string of length $n-i$.
		  \item Perform the marking stage on the resulting, partially marked, string.
	\end{enumerate}
Since $x$ and $y=x+e_i$ agree on all but the $i^{\text{th}}$ coordinate, the
running of the marking stage on $x$ and $y$ in steps 1 and 2 yield the same marking,
and so $\sig(x)$ agrees with $\sig(y)$ in these coordinates, and so,
we may ignore the coordinates marked in the first two steps.

Next we analyze the difference between the markings after the third step.
Denote by $s \in \markedcube[i-1]$ and $t \in \markedcube[n-i]$ the
two partially marked strings such that the resulting strings after the second
step on inputs $x$ and $y$ are $s \circ 0 \circ t$ and $s \circ 1 \circ t$
respectively. 	
Let us suppose for concreteness that the string $s$ contains $a$ unmarked
zeros and $b$ unmarked ones, and the string $t$ contains $c$ unmarked zeros
and $d$ unmarked ones. Recall that at the end of the marking stage, all
unmarked zeros are to the left of all unmarked ones in both $s$ and $t$.
Therefore, we may assume that
\[
	x = 0^a 1^b \circ 0 \circ 0^c 1^d
	\qquad \text{and} \qquad
	y = 0^a 1^b \circ 1 \circ 0^c 1^d.
\]

At this point, it is fairly easy to be convinced that $\dist(\sig(x),\sig(y))$
is bounded by \emph{some} constant. Proving that the constant is $3$
is done by a somewhat tedious case analysis, according to the relations
between $a,b,c$ and $d$.
\begin{claim}\label{claim:case analysis}
    For every $a,b,c,d \in \N$, we have
    \[
        \dist(
        	\sig(0^a 1^b \circ 0 \circ 0^c 1^d),
	        \sig(0^a 1^b \circ 1 \circ 0^c 1^d))
        \leq 3.
    \]
\end{claim}
We postpone the proof of the claim for now,
and move to the ``in particular'' part.
Denote by $U_x = \{i \in [n] \colon \sig(x) = \emptyslot\}$ the unmarked coordinates of $x$.
The first item follows from the fact that $|U_x| = n-2k$,
and similarly $|U_y| = n-2k'$. Therefore, if $\dist(\sig(x),\sig(y)) \leq 3$
it follows that $2|k-k'| \leq |U_x \Delta U_y| \leq 3$,
and hence, since $k$ and $k'$ are integers it follows that $|k-k'|\leq 1$.

For the second item take $c_j \in C_x$ and $c'_{j'} \in C_y$.
Then $c_j$ and $c'_{j'}$ differ in at most 3 coordinates outside $U_x \cap U_y$
since $x$ and $y$ have almost the same signature except at most in three coordinates.
Inside the set $U_x$ the string $c_j$ is just
a sequence of zeros followed by a sequence of ones such that its weight is $j$,
and similarly $c'_{j'}$ restricted to $U_y$ is a sequence of zeros followed by a sequence of ones
such that its weight is $j'$. Hence, inside the set $U_x \cap U_y$ the strings
$c_j$ and $c'_{j'}$ differ in at most $|j-j'|+3$ coordinates.
Therefore, $\dist(c_j,c'_{j'}) \leq |j-j'| + 6$, which completes the proof of
Lemma~\ref{lemma:BTK dist}.
\end{proof}

We now return to the proof of Claim~\ref{claim:case analysis}.
\begin{proof}[Proof of Claim~\ref{claim:case analysis}]
Let $x = 0^a 1^b \circ 0 \circ 0^c 1^d$ and $y = 0^a 1^b \circ 1 \circ 0^c 1^d$.
Our goal is to show that $\dist(\sig(x),\sig(y)) \leq 3$.
The proof uses the following case analysis.

\paragraph{Case 1 ($\boldsymbol{b = c}$).}
In this case we have
\[
	x = 0^a \circ 1^{b} 0^{b} \circ 0 1^d
	\qquad \text{and} \qquad
	y = 0^a 1 \circ 1^{b} 0^{b} \circ 1^d.
\]
Their signatures are
\[
	\sig(x) = \emptyslot^{a} \circ 1^{b} 0^{b} \circ \emptyslot^{d+1}
	\qquad \text{and} \qquad
	\sig(y) = \emptyslot^{a+1} \circ 1^{b} 0^{b} \circ \emptyslot^{d}.
\]
It is easy to verify that in this case the distance
\[
        \dist(
        	\sig(0^a 1^b \circ 0 \circ 0^c 1^d),
	        \sig(0^a 1^b \circ 1 \circ 0^c 1^d))
        \leq 3.
\]

\paragraph{Case 2 ($\boldsymbol{b > c}$).}
In this case we have
\[
	x = 0^a  1^{b-c-1} \circ 1^{c+1} 0^{c+1} \circ 1^d
	\qquad \text{and} \qquad
	y = 0^a 1^{b+1-c} \circ 1^{c} 0^{c} \circ 1^d.
\]
Their signatures are
\[
	\sig(x) = \emptyslot^{a+b-c-1} \circ 1^{c+1} 0^{c+1} \circ \emptyslot^{d}
	\qquad \text{and} \qquad
	\sig(y) = \emptyslot^{a+b-c+1} \circ 1^{c} 0^{c} \circ \emptyslot^{d}.
\]
It is easy to verify that
\[
        \dist(
        	\sig(0^a 1^b \circ 0 \circ 0^c 1^d),
	        \sig(0^a 1^b \circ 1 \circ 0^c 1^d))
        \leq 3.
\]

\paragraph{Case 3 ($\boldsymbol{b < c}$).}
In this case we have
\[
	x = 0^a  \circ 1^{b} 0^{b} \circ 0^{c-b+1} 1^d
	\qquad \text{and} \qquad
	y = 0^a \circ 1^{b+1} 0^{b+1} \circ 0^{c-b-1} 1^d.
\]
Their signatures are
\[
	\sig(x) = \emptyslot^{a} \circ 1^{b} 0^{b} \circ \emptyslot^{d+c-b+1}
	\qquad \text{and} \qquad
	\sig(y) = \emptyslot^{a} \circ 1^{b+1} 0^{b+1} \circ \emptyslot^{d+c-b-1}
\]
It is also easy to verify that in this case the distance is at most
\[
        \dist(
        	\sig(0^a 1^b \circ 0 \circ 0^c 1^d),
	        \sig(0^a 1^b \circ 1 \circ 0^c 1^d))
        \leq 3.
\]

\end{proof}

\subsection{The mapping}\label{sec:the mapping}

In this section we finally prove Theorem~\ref{thm:maj-dict}.
In order to prove the theorem it will be convenient to partition
the hypercube as follows.
For each BTK chain $C$ of the $(n-1)$-dimensional hypercube define
$P_{C} = \{c \circ b \in \cc \colon c \in C, b \in \{0,1\}\}$.
There is a clear one-to-one correspondence between the
$(n-1)$-dimensional BTK chains and our partition of $\cc$.
For example, the block $P_C$ corresponding to the chain $C = \{00,01,11\}$
consists of the following six elements $P_C = \{000,001,010,011,110,111\}$.

\begin{proof}

We define the mapping $\psi$ as follows. Let $n \in \N$ be an odd integer.
For $x \in \cc$, write it as $x = x' \circ x_n$, where $x' \in \cc[n-1]$
represents the first $n-1$ bits of $x$, and $x_n \in \{0,1\}$ is the last bit of $x$.
Let $C = \{ c_k, c_{k+1}, \ldots, c_{n-1-k}\}$ be a symmetric chain
in the BTK partition that contains $x'$,
and let $j$ be the index such that $x' = c_j$.
Define
\begin{equation}\label{eq:definition of f}
    \psi(x = x' \circ x_n) \eqdef \begin{cases}
                    1 \circ c_{2j-(n-k)+x_n}    & \text{if } |x| \geq (n+1)/2; \\
                    0 \circ c_{(n+k)-2j-1-x_n}    & \text{if } |x| \leq (n-1)/2. \\
                \end{cases}
\end{equation}

In order to illustrate the mapping, let us consider as an example the case of $n=3$
and the block $P_C$ that corresponds to the chain $C = \{00,01,11\}$.
\[
  \begin{array}{cccc}
        &\psi(11 \circ 1) & = & 1 \circ 11,\\
        &\psi(11 \circ 0) & = & 1 \circ 01,\\
        &\psi(01 \circ 1) & = & 1 \circ 00,\\
        & --- & - & --- \\
        &\psi(01 \circ 0) & = & 0 \circ 00,\\
        &\psi(00 \circ 1) & = & 0 \circ 01,\\
        &\psi(00 \circ 0) & = & 0 \circ 11,\\
  \end{array}
\]

Note that $\psi$ maps the upper half of $\cc$ to points whose first coordinate is $1$,
and maps the lower half of $\cc$ to points whose first coordinate is $0$.
It should be mentioned that the mapping $\psi$ restricted to the upper half of $\cc$
is exactly the mapping used in~\cite{BCS-bilip}, where it was shown that
the restriction of $\psi$ to the upper half of $\cc$ is a bi-Lipschitz bijection.
Similarly, the restriction of $\psi$ to the lower half of $\cc$
is also a bi-Lipschitz bijection, and so, as mentioned above,
the main difficulty in this construction was to ``stitch'' these two bijections so that
the endpoints of the edges in the middle layer are mapped by $\psi$ close to each other.

We next show that the mapping $\psi$ is indeed
a $11$-Lipschitz bijection from $\Maj$ to $\Dict$.
Note that by the triangle inequality it is enough to show that for
all edges $(x,y=x+e_i) \in \cc$ it holds that $\dist(\psi(x),\dist(y)) \leq 11$.

Write $x = x' \circ x_n$ where $x' \in \cc[n-1]$
represents the first $n-1$ bits of $x$, and $x_n \in \{0,1\}$ is the last bit of $x$.
Analogously write $y = y' \circ y_n$.
Let $C_{x'} = \{c_k,c_{k+1}\dots,c_{n-k} \}$ and
$C_{y'}  = \{c'_{k'},c_{k'+1}\dots,c'_{n-k'} \}$
be the BTK chains that contain $x' = c_j$ and $y' = c'_{j'}$ respectively,
where $j = |x'|$ and $j' = |y'|$.
Recall that $|j-j'| \leq 1$ since $|x'| = j$ and $|y'| = j'$ and $x \sim y$.

Our goal is to show that $\dist( \psi(x), \psi(y) ) \leq 10$.
We now consider several cases.

\paragraph{Case 1  ($\boldsymbol{|x| = (n-1)/2}$ and $\boldsymbol{|y| = (n+1)/2}$).}
In this case we have
\[
	\psi(x) = 0 \circ c_{(n+k)-(n-1)-1-x_n} = 0 \circ c_{k-x_n}
\]
and
\[
	\psi(y) = 1 \circ c_{n+1 - (n-k') + y_n} =  1 \circ c'_{k'+1+y_n}.
\]
By Lemma~\ref{lemma:BTK dist} we have $|k-k'| \leq 1$ and hence
$\dist(\psi(x),\psi(y)) \leq 1 + |(k-x_n) - (k'+1+y_n)| + 6 \leq 11$,
as required.

\paragraph{Case 2 ($\boldsymbol{|x| \geq (n+1)/2}$ and $\boldsymbol{|y| \geq (n+1)/2}$).}
In this case we have
$|(2j-(n-k)+x_n) - (2j'-(n-k')+y_n) | \leq 5$,
and so by Lemma~\ref{lemma:BTK dist} the distance between
$\psi(x)$ and $\psi(y)$ is at most $\dist(\psi(x),\psi(y))  = \dist(c_{2j-n+k+x_n}, c'_{2j'-n+k'+y_n})\leq 5 + 6 = 11$.

\paragraph{Case 3  ($\boldsymbol{|x| \leq (n-1)/2}$ and $\boldsymbol{|y| \leq (n-1)/2}$).}
This is handled similarly to case 2.

\medskip
This completes the proof of Theorem~\ref{thm:maj-dict}.
\end{proof}

\section{A Linear Bijection from $\Dict $ to $\Parity$}\label{sec:dict-to-xor}

In this section we prove Theorem~\ref{thm:dict-xor}.

\begin{proof}[Proof of Theorem~\ref{thm:dict-xor}]
We give an explicit mapping from $\Dict$ to $\Parity$ that is a linear transformation over $GF(2)$. Let $A$ be the matrix representing the linear transformation. We first show that the mapping satisfies the conditions in the lemma if $A$ has the following properties.

\begin{enumerate}
\item $A$ is invertible.

\item The first column of $A$ has odd weight, and all other columns have even weight.

\item All rows of $A$ have weight $3$.

\item All columns of $A$ have weight $2$.

\item All columns of $A^{-1}$ have weight $O(\log(n))$.

\end{enumerate}

The first condition implies that $A$ is a bijection. The second condition implies that $A$ maps from $\Dict$ to $\Parity$. To see this, consider $\Parity (Av)$ for any vector $v \in \{0, 1\}^n$. This is the sum over $GF(2)$ of the weights of all columns $j$ for which $v_{j} = 1$. Because the weights of all columns but the first are $0$, they can be ignored, and therefore $\Parity(Av) = 1$ if and only if $v_1 = 1$.

The third condition implies that each output bit is local, as the $i$th output bit depends only on the $i$th row of $A$.
The fourth condition implies that $A$ is $2$-Lipschitz. To see this, note that $\dist(x, y)$ is the weight of $x-y$, and $\dist(Ax, Ay)$ is the weight of $A(x-y)$. If the weight of each column of $A$ is at most $C$, then the weight of $A(x-y)$ is at most $C$ times the weight of $(x-y)$.
The same argument applied to $A^{-1}$ implies that $A^{-1}$ is $O(\log(n))$-Lipschitz.

Note that under the assumption that the mapping is a linear transformation, the above is necessary for a mapping to satisfy the conditions of the lemma.

We now construct the mapping $A$. Let $G$ be the complete binary tree on $n$ vertices,
with directed edges so that each points to the child. We uniquely label each vertex with a label from $1$ to $n$, with the root labeled $1$.
Let $A$ be $I_n+M$ where $M$ is the adjacency matrix of $G$. That is $A_{i,j} = 1$ if $i$ is the parent of $j$ in $G$ or $i=j$, and $A_{i,j} = 0$ otherwise.

We claim that the inverse of $A$ is the matrix defined as $B_{i,j} = 1$ if $j$ is a  descendant of $i$, and $B_{i,j} = 0$ otherwise.
Indeed, consider the $(i, j)$th entry of the product $A \cdot B$.
Then
\begin{eqnarray*}
	(A \cdot B)_{i,j} & = & \sum_{k=1}^n A_{i,k} B_{k,j} \\
	 & = & | \{ k : \text{$i=k$ or $i$ is the parent of $k$} \} \cap \{ k : \text{$k=j$ or $j$ is a descendant of $k$}\} |
\end{eqnarray*}

If $i = j$, these sets have exactly one element in common, $i$, and therefore $(A \cdot B)_{i, j} = 1$.
If vertex $j$ is not in the subtree rooted at vertex $i$, these sets have no vertices in common, and therefore $(A \cdot B)_{i, j} = 0$.
Finally, if vertex $j$ is in the subtree rooted at $i$ but is not $i$, these sets have exactly two vertices in common, and therefore $(A \cdot B)_{i, j} =0$.

Because vertex $1$ is the only vertex without a parent, the first column of $A$ has weight $1$.
All other vertices have exactly one parent, and hence all other columns of $A$ have weight $2$.
The weight of each row of $A$ is at most $3$ since this is equal to the number of children of the corresponding vertex plus 1.
The $i$th column of $A^{-1}$ is the indicator vector of the set of ancestors of $i$ in $G$ including $i$ itself.
The size of this set is bounded above by $\log(n)+1$. Therefore, $A$ satisfies the conditions of the lemma.
\end{proof}

Note that the above proof can be generalized to obtain a mapping $\phi$ that is
$L$-local $2$-Lipschitz such that $\phi^{-1}$ is $C$-Lipschitz,
for any $L$ and $C$ that satisfy $(L-1)^C \geq n$ by replacing the tree $G$ in the proof
with a complete $L-1$-ary tree on $n$ vertices. Such a tree will have height less than $C$.

\section{Any Bijection from $\Parity$ to $\Maj$ has Large Average Stretch}\label{sec:xor-to-maj}

In this section we prove Theorem~\ref{thm:xor-maj}.
We start with the second part of the theorem.

\begin{proposition}
    There exists a $C$-Lipschitz bijection $\psi \colon \cc \to \cc$ from $\Maj$ to $\Parity$
    for some absolute constant $C$.
\end{proposition}

\begin{proof}
	Take the $C$-Lipschitz bijection from $\Maj$ to $\Dict$ from
	Theorem~\ref{thm:maj-dict}, and compose it with the 2-Lipschitz
	bijection $\phi(x_1,x_2,\dots,x_n) = (x_1+x_2,x_2+x_3,\dots,x_{n-1}+x_n,x_n)$
	from $\Dict$ to $\Parity$.
	The resulting mapping is clearly a $2C$-Lipschitz bijection from $\Maj$ to $\Parity$.
\end{proof}

Next we prove the first part of Theorem~\ref{thm:xor-maj}
showing that any bijection from $\Parity$ to $\Maj$ must have large average stretch.
In fact we prove a stronger statement,
saying that in every direction $i \in [n]$ it holds that
the average stretch in the direction $e_i$ must be $\Omega(\sqrt{n})$.

\begin{proposition}
    Let $\phi$ be a bijection from $\Parity$ to $\Maj$, and let $i \in [n]$.
    Then $\E_{x}[\dist(\phi(x),\phi(x+e_i)] \geq c\sqrt{n}$ for some absolute constant $c>0$.
\end{proposition}

\begin{proof}
	Fix $i \in [n]$ and consider all edges in the direction $e_i$, i.e.,
	the edges of the form $\{ (x,x+e_i) \}_{x \in \cc}$.
	Since $\Parity(x) \neq \Parity(x+e_i)$ for all $x \in \cc$
	it follows that for every such edge one of its endpoints must be mapped
	to the upper half of the hypercube, and the other endpoint to the
	bottom half of the hypercube.
	On the other hand, since each level of the hypercube contains at most ${n \choose n/2}<2^n/\sqrt{n}$ vertices,
	it follows that
	for $0.9$-fraction of points $z \in \cc$
	their weight differs from $n/2$ by more than $0.01\sqrt{n}$.
	Let us call such $z$ \emph{typical}.
	Therefore, for $0.8$ of inputs $x$ it holds that at least one of the
	endpoints of the edge $(x,x+e_i)$ is mapped to a typical point.
	Let us say for concreteness that $x$ is typical and is
	mapped above level $n/2 + 0.01\sqrt{n}$. Then, since $\phi(x) \neq \phi(x + e_i)$
	it follows that $\phi(x+e_i)$ must belong to the lower half of the hypercube,
	and thus $\dist(\phi(x),\phi(x+e_i)) \geq \Omega(\sqrt{n})$.
	Therefore, for at least $0.8$ fraction of the edges in the direction $e_i$
	it holds that $\dist(\phi(x),\phi(x+e_i)) \geq 0.02\sqrt{n}$.
\end{proof}

This clearly implies Theorem~\ref{thm:xor-maj}.

\section{Bijection from $\Dict$ to a Random Function}\label{sec:dict-to-rand}

In this section we prove Theorem~\ref{thm:rand}.
We start with the following claim.

\begin{lemma}\label{lemma:rand HLN}
Let $A \seq \{0,1\}^n$ be a random set chosen by picking each $x \in \cc$
to be in $A$ independently with probability $0.5$.
Then with probability $1-2^{-2^{\Omega(n)}}$ there exists an injective mapping $\phi_A \colon \cc[n-1] \to \cc$
such that the following holds.

\begin{enumerate}
\item For all $x \in \cc[n-1]$ it holds that $\dist(x,\phi_A(x)_{[1,\dots,n-1]}) \leq 1$,
where $\phi_A(x)_{[1,\dots,n-1]}$ denotes the restriction of $\phi_A(x)$ to the first $n-1$ coordinates.
\item $\Pr_{x \in \cc[n-1]}[\phi_A(x) \in A] = 1-O(1/n)$.
\end{enumerate}

\end{lemma}

We postpone the proof until later and show how to prove Theorem~\ref{thm:rand}
using Lemma~\ref{lemma:rand HLN}.

\begin{proof}[Proof of Theorem~\ref{thm:rand}]
    We start with the following simple claim, which is
    immediate from Lemma~\ref{lemma:rand HLN}.
    \begin{claim}\label{claim:rand bal set}
        Let $A \subset \cc$ be a uniformly random subset of size $2^{n-1}$.
        Then, with probability $1-2^{-2^{\Omega(n)}}$
        there exists a bijection $\phi \colon \cc[n-1] \to A$
        such that for $1 - O(1/n)$ fraction of the inputs it holds that
        $\dist(x,\phi(x)_{[1,\dots,n-1]}) \leq 1$.
    \end{claim}
    \begin{proof}
        Let us sample a subset $A_1 \subset \cc$ of
        size exactly $2^{n-1}$ in the following manner.
        Pick a random subset $A \seq \{0,1\}^n$ by choosing each $x \in \cc$
        to be in $A$ independently with probability $0.5$.
        Then, if $A < 2^{n-1}$ we add to $A$ uniformly random elements
        from $\cc \setminus A$ one by one until the size of $A$
        becomes $2^{n-1}$.
        Similarly, if $A > 2^{n-1}$ we remove random elements
        from $A$ one by one until the size of $A$ becomes $2^{n-1}$.
        Let $A_1$ be the obtained set.
        Clearly $A_1$ is indeed a uniformly random subset of $\cc$ of size $2^{n-1}$.

    	By Lemma~\ref{lemma:rand HLN}, with probability $1-2^{-2^{\Omega(n)}}$
        there is an injective mapping $\phi \colon \cc[n-1] \to A$
        such that $\dist(x,\phi(x)_{[1,\dots,n-1]}) \leq 1$ for all $x \in \cc[n-1]$,
        and for all but $O(1/n)$ fraction of the inputs it holds that $\phi(x) \in A_1$.
        By the Chernoff bound with probability $1-2^{-2^{\Omega(n)}}$
        the sets $A$ and $A_1$ differ in at most most $2^n/n$ elements.
        Therefore, we can modify $\phi$ in $O(1/n)$ fraction of the
        inputs so that the obtained mapping is a bijection from $\cc[n-1]$ to $A$
        that satisfies the requirements of the claim.

    \end{proof}

    In order to prove Theorem~\ref{thm:rand} we sample a uniformly random balanced
    boolean function $f \colon \{0,1\}^n \to \{0,1\}$ as follows.
    Pick $A_1 \subset \cc$ of size $2^{n-1}$ uniformly at random,
    and let $A_0 = \cc \setminus A_1$,
    Define $f$ to be the indicator function of $A_1$,
    i.e., $A_1 = f^{-1}(1)$ and $A_0 = f^{-1}(0)$.

	Since $A_1$ is a uniformly random set of size $2^{n-1}$
	by Claim~\ref{claim:rand bal set} with probability $1-2^{-2^{\Omega(n)}}$
	there is a bijection $\phi_1 \colon \{ x \in \cc[n-1] \colon x_1 = 1 \} \to A_1$
	such that $\dist(x,\phi_1(x)) \leq 2$ for all but $O(1/n)$ fraction of the domain of $\phi_1$.
	Similarly, $A_0$ is also a uniformly random subset of $\cc$ of size $2^{n-1}$,
	and hence with probability $1-2^{-2^{\Omega(n)}}$
	there is a bijection $\phi_0 \colon \{ x \in \cc[n-1] \colon x_1 = 0 \} \to A_0$
	such that $\dist(x,\phi_0(x)) \leq 2$ for all but $O(1/n)$ fraction of the inputs
	
	By the union bound with probability $1-2^{-2^{\Omega(n)}}$
	both $\phi_0$ and $\phi_1$ exist.
	Since $\phi_0$ and $\phi_1$ are defined on disjoint domains, whose union is the entire hypercube,
	we can define $\phi \colon \cc \to \cc$ to be
	\[
		\phi(x) =
			\begin{cases}
				\phi_1(x)	& \text{if } x_1 = 1 \\
				\phi_0(x)	& \text{if } x_1 = 0.
			\end{cases}
	\]
	Clearly $\phi$ is a bijection from $\Dict$ to $f$ and it satisfies
	\begin{equation}\label{eq:rand dist2}
		\Pr_{x \in \cc}[\dist(x,\phi(x)) \leq 2] = 1 - O(1/n),
	\end{equation}
	as required. The ``in particular'' part of Theorem~\ref{thm:rand}
	follows immediately from~\eqref{eq:rand dist2}.
	Indeed, by~\eqref{eq:rand dist2} it follows that
	$1 - O(1/n)$ of the edges $(x,x+e_i)$ satisfy
	\[
		\dist(\phi(x),\phi(x+e_i)) \leq \dist(\phi(x),x) + \dist(x, x+e_i) + \dist(x + e_i,\phi(x+e_i)) \leq 5,
	\]
	and therefore
	\[
		\avgstretch(\phi) =\E_{x \in \cc, i \in n}[\dist(\phi(x),\phi(x+e_i)] \leq 5 \cdot (1 - O(1/n)) + n \cdot O(1/n) = O(1).
	\]
	In order to see that $\avgstretch(\phi^{-1}) = O(1)$ note that
	$\phi$ is a bijection, and so by~\eqref{eq:rand dist2}
	we have $\Pr_{x \in \cc}[\dist(x,\phi^{-1}(x)) \leq 2] = 1 - O(1/n)$.
	This completes the proof of Theorem~\ref{thm:rand}.
\end{proof}
	
We now return to the proof of Lemma~\ref{lemma:rand HLN}.
The proof relies on an algorithm from~\cite{HLN87}.

\begin{proof}[Proof of Lemma~\ref{lemma:rand HLN}]
In order to describe the algorithm let $A$ be a random subset of $\cc$. For each
$x \in \cc[n-1]$ say that $x$ is \emph{rich} if both $x \circ 0$ and $x \circ 1$
belong to $A$, and say that $x$ is \emph{poor} if none of $x \circ 0$ and $x \circ 1$ belongs to $A$.
If there were no poor vertices in $\cc[n-1]$
then we could define $\phi_A$ by extending its input $x$
to either $x \circ 0$ or $x \circ 1$.
However, since the subset $A$ is uniformly random, roughly $1/4$ fraction
of the vertices in $\cc[n-1]$ will be poor, and we will match all but $O(1/n)$ fraction
of poor vertices with a neighboring rich vertex. Then, we will define a mapping $\phi_A$
in the following way:
(1) if $x$ is neither rich nor poor, then define
$\phi_A(x) = x \circ b$, where $b \in \{0,1\}$ is such that $x \circ b \in A$,
(2) if $x$ is rich, then define $\phi_A(x) = x \circ 1$,
(3) if $x$ is poor and is matched with a rich vertex $y$, then define $\phi_A(x) = y \circ 1$,
(4) otherwise, $x$ is poor and is not matched with a rich vertex, in which case we define $\phi_A(x) = x \circ 0$.

Clearly such a mapping $\phi_A$ satisfies the condition that $\dist(x,\phi_A(x)_{[1,\dots,n-1]}) \leq 1$
for all $x \in \cc[n-1]$.
We will define the matching so that only $O(1/n)$ fraction of the poor vertices will be of type (4),
i.e., will be poor and not matched to a neighboring rich vertex, and hence only those vertices $x$
will be so that $\phi_A(x) \notin A$.

The algorithm for finding such a matching is the following.

    \begin{algorithm}
    \algsetup{indent=2em}
    \caption{Matching poor vertices to rich vertices}\label{alg:HLN}
    \begin{algorithmic}[1]
            \FOR{$i =1 \dots, n/2$}
                \IF{$x$ is poor and not matched and $x+e_i$ is rich and not matched }
                    \STATE Match $x$ with $x+e_i$
                \ENDIF
            \ENDFOR
    \end{algorithmic}
    \end{algorithm}

\begin{remark}
	We remark that we could allow the loop to run until $n$, however for the analysis it will be more convenient to stop after $n/2$ steps.
\end{remark}

The following two claims from~\cite{HLN87} are the key steps in the analysis of the algorithm above.
\begin{claim}[{\cite[Lemma 1]{HLN87}}]\label{claim:indHLN}
	For every $k \leq n/2$, the status of $x$ in the $k$th iteration of the algorithm
	is independent of all vertices that differ from $x$ in some coordinate larger than $k$.

	In particular for any $z \in \cc[k]$ let $A_z = \{x \in \cc[n-1] \colon x_{[1,\dots,k]} = z\}$.
	Then, in the $k$th iteration of the algorithm the status of each vertex in $A_z$ is independent of the others.
\end{claim}

\begin{proof}
	At each iteration $i$ the vertices that affect each other are matched according to the edges in the direction $e_i$.
	Therefore, any two vertices that differ in some coordinate larger than $k$
	had no interaction between them, and so are independent of each other.
\end{proof}

\begin{claim}\label{claim:recursion}
	For every $x \in \cc[n-1]$ and $i = 1,\dots, n/2$
	let	$p_i$ be the probability that $x$ is poor and unmatched after iteration $i$.
	Then
	\begin{enumerate}
	\item $p_{i+1} = p_i(1-p_i)$ for all $i < n/2$.
	\item $p_{n/2} < 2/n$.
	\end{enumerate}		
\end{claim}

\begin{proof}

Let $q_i$ be the probability that $x$ is rich but unmatched after round $i$, analogous to $p_i$.  In the $(i+1)$st round, $x$ is poor and unmatched if $x$ was poor and unmatched after the $i$th round, and $x+e_{i+1}$ is rich and unmatched after the $i$th round.  By Claim~\ref{claim:indHLN}, these two events are independent, and therefore $p_{i+1} = p_i(1-q_i).$  Similarly, $q_{i+1}$ can be expressed as $q_{i+1} = q_i(1-p_i)$.  Subtracting these two equations, we see that $q_{i+1}-p_{i+1} = q_i-p_i = q_0-p_0$ for all $i$.  This is natural as the difference between the number of rich unmatched vertices and poor unmatched vertices stays constant throughout the rounds.  Because $p_0 = 1/4$ and $q_0 = 1/4$, this difference is $0$ and therefore $q_i = p_i$ for all $i$.  Substituting $q_i$ in the expression for $p_{i+1}$ yields $p_{i+1} = p_i(1-p_i)$.  This proves the first part of the claim.

To prove the second part of the claim,
we show by induction that $p_i \leq 1/i$ for all $i \geq 1$.
Indeed, since $p_0 = 1/4$ the claim holds for $i \leq 4$.
For the induction step for $i \geq 4$, if $p_i < 1/i \leq 1/4$, then
$p_{i+1} = p_i(1 - p_i) \leq (1/i)(1 - 1/i) < 1/(i+1)$, as required.

\end{proof}

	We are now ready to complete the proof of Lemma~\ref{lemma:rand HLN}.
	For each $z \in \cc[n/2]$ consider the set $A_z = \{x \in \cc[n-1] \colon x_{[1,\dots,k]} = z\}$.
	By Claims~\ref{claim:indHLN} and~\ref{claim:recursion}
	each $x \in A_z$ is poor and unmatched with probability $p_{n/2} < C/n$
	independently of all other vertices in $A_z$.
	Therefore, since $|A_z| = 2^{n/2}$, by the Chernoff bound the probability that $A_z$ contains more than a $2C/n$ fraction of
	poor and unmatched vertices is at most $2^{-2^{\Omega(n)}}$.
	By taking union bound over all $z \in \cc[n/2]$ we conclude that
	with probability $1- 2^{-2^{\Omega(n)}}$ there exists a
	matching that matches all but $O(1/n)$ fraction of the proof vertices with a rich neighboring vertex.
	This completes the proof of Lemma~\ref{lemma:rand HLN}.
\end{proof}

\section{Open Problems}

Below we list several open problems.

\begin{question}
    In Theorem~\ref{thm:dict-xor} we constructed a linear mapping from
    $\Dict$ to $\Parity$ that is 3-local, 2-Lipschitz
    such that its inverse is $O(\log(n))$-Lipschitz.
    Is there a mapping from
    $\Dict$ to $\Parity$ that is $O(1)$-local, and $O(1)$-bi-Lipschitz?
    In particular, it would be interesting to find
    such a mapping that is non-linear.
\end{question}

\begin{question}
    We proved in Theorem~\ref{thm:rand}
    that for a random balanced function $f$
    with high probability there is a mapping $\phi_f$
    from $\Dict$ to $f$ such that
    $\avgstretch(\phi) = O(1)$ and $\avgstretch(\phi^{-1}) = O(1)$.
    Is it true that with high probability
    there is a bi-Lipschitz mapping from $\Dict$ to a random function,
    i.e., a mapping with bounded worse case stretch?
\end{question}

\begin{question}
    We proved in Theorem~\ref{thm:xor-maj}
    that any mapping from $\Parity$ to $\Maj$ must
    have average stretch larger than $\Omega(\sqrt{n})$.
    Is this bound tight?
    Is there a mapping $\phi$ from $\Parity$ to $\Maj$
    such that $\avgstretch(\phi) = o(n)$?
\end{question}

In this paper we only considered mappings between functions
with the same domain. If we allow one of the functions
to have a larger domain, we may relax the requirement
that a mapping between function must be a bijection,
and only require that the mapping be one-to-one.
Given Theorem~\ref{thm:xor-maj} we ask the following question.

\begin{question}
    Is there a Lipschitz embedding $\phi : \cc \to \cc[\poly(n)]$
    such that $\Parity(z) = \Maj(\phi(z))$ for all $z \in \cc$?
\end{question}

%

\bibliographystyle{alpha}
\bibliography{f-embed}

\end{document}